\documentclass{stacs_proc}
\usepackage[T1]{fontenc}
\usepackage[latin1]{inputenc}
\usepackage{verbatim}
\usepackage{float}
\usepackage{graphicx}
\usepackage{amssymb}

\makeatletter

\newenvironment{proofsketch}{\noindent{\it Proof Sketch.}\hspace*{0.75em}}{\qed\medskip}

\def\R{\mathbb{R}}

\title{Geodesic Fréchet Distance Inside a Simple Polygon}
\thanks{The full version of this paper is available as a technical report \cite{Cook2007}.}
\author[]{A.F. Cook}{Atlas F. Cook IV}%
\author[]{C. Wenk}{Carola Wenk}%
\address[]{Department of Computer Science, University of Texas at San Antonio\newline One UTSA Circle, San Antonio, TX 78249-0667}%
\email{{acook,carola}@cs.utsa.edu}%
\thanks{This work has been supported by the National Science Foundation
   grant NSF CAREER CCF-0643597.}%
\subjclass{Computational Geometry}%

\makeatother

\makeatother

\begin{document}
\begin{abstract}
We unveil an alluring alternative to parametric search that applies
to both the non-geodesic and geodesic Fréchet optimization problems.
This randomized approach is based on a variant of red-blue intersections
and is appealing due to its elegance and practical efficiency when
compared to parametric search.

We present the first algorithm for the geodesic Fréchet distance between
two polygonal curves $A$ and $B$ inside a simple bounding polygon
$P$. The geodesic Fréchet decision problem is solved almost as fast
as its non-geodesic sibling and requires $O(N^{2}\log k)$ time and
$O(k+N)$ space after $O(k)$ preprocessing, where $N$ is the larger
of the complexities of $A$ and $B$ and $k$ is the complexity of
$P$. The geodesic Fréchet optimization problem is solved by a randomized
approach in $O(k+N^{2}\log kN\log N)$ expected time and $O(k+N^{2})$
space. This runtime is only a logarithmic factor larger than the standard
non-geodesic Fréchet algorithm \cite{Alt1995}. Results are also presented
for the geodesic Fréchet distance in a polygonal domain with obstacles
and the geodesic Hausdorff distance for sets of points or sets of
line segments inside a simple polygon $P$. 
\end{abstract}
\maketitle

\stacsheading{2008}{193-204}{Bordeaux}
\firstpageno{193}

\section{Introduction\label{sec:Introduction}}



The comparison of geometric shapes is essential in various applications
including computer vision, computer aided design, robotics, medical
imaging, and drug design. The Fréchet distance is a similarity metric
for continuous shapes such as curves or surfaces which is defined
using reparametrizations of the shapes. Since it takes the continuity
of the shapes into account, it is generally a more appropriate distance
measure than the often used Hausdorff distance. 
The Fréchet distance for curves is commonly illustrated by a person
walking a dog on a leash \cite{Alt1995}. The person walks forward
on one curve, and the dog walks forward on the other curve. As the
person and dog move along their respective curves, a leash is maintained
to keep track of the separation between them. The Fréchet distance
is the length of the \emph{shortest} leash that makes it possible
for the person and dog to walk from beginning to end on their respective
curves without breaking the leash. See section \ref{sec:Preliminaries}
for a formal definition of the Fréchet distance.


Most previous work assumes an obstacle-free environment where the
leash connecting the person to the dog has its length defined by an
$L_{p}$ metric. In \cite{Alt1995} the Fréchet distance between polygonal
curves $A$ and $B$ is computed in arbitrary dimensions for obstacle-free
environments in $O(N^{2}\log N)$ time, where $N$ is the larger of
the complexities of $A$ and $B$. Rote \cite{Rote2005} computes
the Fréchet distance between piecewise smooth curves. Buchin et al.\ \cite{Buchin2006}
show how to compute the Fréchet distance between two simple polygons.
Fréchet distance has also been used successfully in the practical
realm of map matching \cite{Wenk2006}. All these works assume a leash
length that is defined by an $L_{p}$ metric.

This paper's contribution is to measure the leash length by its geodesic
distance inside a simple polygon $P$ (instead of by its $L_{p}$
distance). To our knowledge, there are only two other works that employ
such a leash. One is a workshop article \cite{Maheshwari2005} that
computes the Fréchet distance for polygonal curves $A$ and $B$ on
the surface of a convex polyhedron in $O(N^{3}k^{4}\log(kN))$ time.
The other paper \cite{Efrat2002} applies the Fréchet distance to
morphing by considering the polygonal curves $A$ and $B$ to be obstacles
that the leash must go around. Their method works in $O(N^{2}\log^{2}N)$
time but only applies when $A$ and $B$ both lie on the boundary
of a simple polygon. Our work can handle both this case and more general
cases. We consider a simple polygon $P$ to be the only obstacle and
the curves, which may intersect each other or self-intersect, both
lie inside $P$.


A core insight of this paper is that the free space in a geodesic
cell (see section \ref{sec:Preliminaries}) is $x$-monotone, $y$-monotone,
and connected. We show how to quickly compute a cell boundary and
how to propagate reachability through a cell in constant time. This
is sufficient to solve the geodesic Fréchet decision problem. To solve
the geodesic Fréchet optimization problem, we replace the standard
parametric search approach by a novel and asymptotically faster (in
the expected case) randomized algorithm that is based on red-blue
intersection counting. We show that the geodesic Fréchet distance
between two polygonal curves inside a simple bounding polygon can
be computed in $O(k+N^{2}\log kN\log N)$ expected time and $O(k+N^{3}\log kN)$
worst-case time, where $N$ is the larger of the complexities of $A$
and $B$ and $k$ is the complexity of the simple polygon. The expected
runtime is almost a quadratic factor in $k$ faster than the straightforward
approach, similar to \cite{Efrat2002}, of partitioning each cell
into $O(k^{2})$ subcells. Briefly, these subcells are simple combinatorial
regions based on \emph{pairs} of hourglass intervals. It is notable
that the randomized algorithm also applies to the non-geodesic Fréchet
distance in arbitrary dimensions. We also present algorithms to compute
the geodesic Fréchet distance in a polygonal domain with obstacles
and the geodesic Hausdorff distance for sets of points or sets of
line segments inside a simple polygon.

\section{Preliminaries\label{sec:Preliminaries}}

Let $k$ be the complexity of a simple polygon $P$ that contains
polygonal curves $A$ and $B$ in its interior. In general, a \emph{geodesic}
is a path that avoids all obstacles and cannot be shortened by slight
perturbations \cite{Mitchell1987}. However, a geodesic inside a simple
polygon is simply a unique shortest path between two points. Let $\pi(a,b)$
denote the geodesic inside $P$ between points $a$ and $b$. The
\emph{geodesic distance} $d(a,b)$ is the length of a shortest path
between $a$ and $b$ that avoids all obstacles, where length is measured
by $L_{2}$ distance.

Let $\downarrow$, $\uparrow$, and $\downarrow\uparrow$ denote decreasing,
increasing, and decreasing then increasing functions, respectively.
For example, {}``$H$ is $\downarrow\uparrow$-bitonic'' means that
$H$ is a function that decreases monotonically then increases monotonically.
A \emph{bitonic} function has at most one change in monotonicity.


\noindent The Fréchet distance for two curves $A,B:[0,1]\rightarrow\R^{l}$
is defined as\[
\delta_{F}(A,B)=\inf_{f,g:[0,1]\rightarrow[0,1]}\sup_{t\in[0,1]}\ d'(\ A(f(t)),B(g(t))\ )\]
 where $f$ and $g$ range over continuous non-decreasing reparametrizations
and $d'$ is a distance metric for points, usually the $L_{2}$ distance,
and in our setting the geodesic distance. For a given $\varepsilon>0$
the {\em free space} is defined as $FS_{\varepsilon}(A,B)=\{(s,t)\;|\; d'(A(s),B(t))\leq\varepsilon\}\subseteq[0,1]^{2}$.
A free space cell $C\subseteq[0,1]^{2}$ is the parameter space defined
by two line segments $\overline{ab}\in A$ and $\overline{cd}\in B$,
and the free space inside the cell is $FS_{\varepsilon}(\overline{ab},\overline{cd})=FS_{\varepsilon}(A,B)\cap C$.

The decision problem to check whether the Fréchet distance is at most
a given $\varepsilon>0$ is solved by Alt and Godau \cite{Alt1995}
using a \emph{free space diagram} which consists of all free space
cells for all pairs of line segments of $A$ and $B$. Their dynamic
programming algorithm checks for the existence of a monotone path
in the free space from $(0,0)$ to $(1,1)$ by propagating \emph{reachability
information} cell by cell through the free space.

\subsection{Funnels and Hourglasses\label{sub:Funnels-and-Hourglasses}}

Geodesics in a free space cell $C$ can be described by either the
funnel or hourglass structure of \cite{Guibas1986}. A funnel describes
all shortest paths between a point and a line segment, so it represents
a horizontal (or vertical) line segment in $C$. An hourglass describes
all shortest paths between two line segments and represents all distances
in $C$.

The \emph{funnel} $\mathcal{F}_{p,\overline{cd}}$ describes all shortest
paths between an apex point $p$ and a line segment $\overline{cd}$.
The boundary of $\mathcal{F}_{p,\overline{cd}}$ is the union of the
line segment $\overline{cd}$ and the shortest path chains $\pi(p,c)$
and $\pi(p,d)$. The \emph{hourglass} $\mathcal{H}_{\overline{ab},\overline{cd}}$
describes all shortest paths between two line segments $\overline{ab}$
and $\overline{cd}$. The boundary of $\mathcal{H}_{\overline{ab},\overline{cd}}$
is composed of the two line segments $\overline{ab}$, $\overline{cd}$
and at most four shortest path chains involving $a$, $b$, $c$,
and $d$. See Figure \ref{fig:hourglass-to-pseudo-funnel}. Funnel and hourglass
boundaries have $O(k)$ complexity because shortest paths inside a
simple polygon $P$ are acyclic, polygonal, and only have corners
at vertices of $P$ \cite{Guibas1987}.


Any horizontal or vertical line segment in a geodesic free space cell
is associated with a funnel's distance function $F_{p,\ \overline{cd}}:[c,d]\rightarrow\mathbb{R}$
with $F_{p,\ \overline{cd}}(q)=d(p,q)$. The below three results are
generalizations of Euclidean properties and are omitted.%
{} See \cite{Cook2007} for details.

\begin{lemma}\label{lemma-rudolph}

$F_{p,\ \overline{cd}}$ is $\downarrow\uparrow$-bitonic.

\end{lemma}

\begin{corollary}\label{corollary-rudolph}

Any horizontal (or vertical) line segment in a free space cell has
at most one connected set of free space values.

\end{corollary} 

Consider the hourglass $\mathcal{H}_{\overline{ab},\ \overline{cd}}$
in Figure \ref{fig:hourglass-to-pseudo-funnel}. Let the \emph{shortest}
distance from $a$ to any point on $\overline{cd}$ occur at $M_{a}\in\overline{cd}$.
Define $M_{b}$ similarly. As $p$ varies from $a$ to $b$, the \emph{minimum}
distance from $p$ to $\overline{cd}$ traces out a function $H_{\overline{ab},\ \overline{cd}}:[a,b]\rightarrow\mathbb{R}$
with $H_{\overline{ab},\ \overline{cd}}(p)=\min_{q\in[c,d]}d(p,q)$.%
\begin{figure}
\begin{centering}
\includegraphics[scale=0.33]{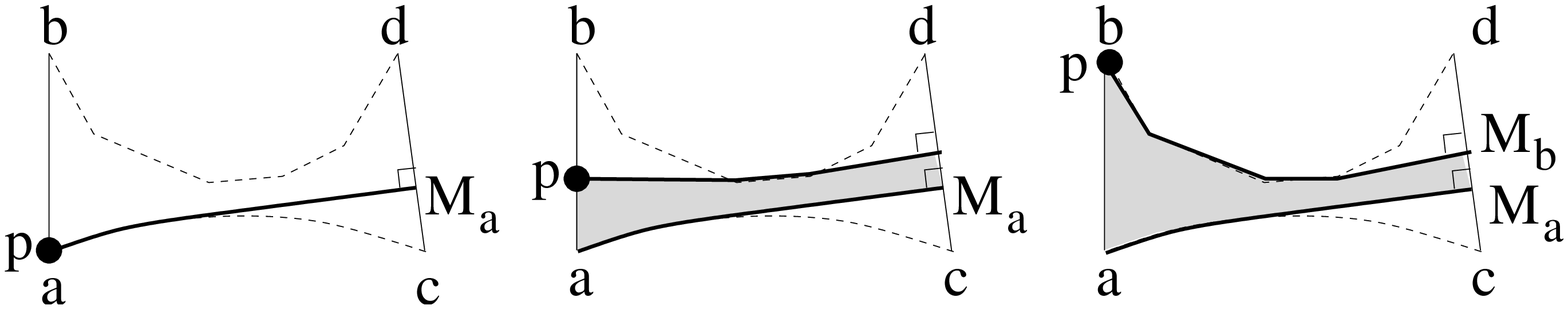} 
\par\end{centering}

\caption{Shortest paths in the hourglass $\mathcal{H}_{\overline{ab},\overline{cd}}$
define $H_{\overline{ab},\ \overline{cd}}$.  \label{fig:hourglass-to-pseudo-funnel}}

\end{figure}

\begin{lemma} \label{lemma-hourglass-bitonicity}

$H_{\overline{ab},\ \overline{cd}}$ is $\downarrow\uparrow$-bitonic.

\end{lemma}

\section{Geodesic Cell Properties\label{sec:Geodesic-Cell-Properties}}

Consider a geodesic free space cell \emph{C} for polygonal curves
$A$ and \emph{$B$} inside a simple polygon. Let $\overline{ab}\in A$
and $\overline{cd}\in B$ be the two line segments defining $C$.

\begin{lemma} \label{one-free-region-per-cell}

For any $\varepsilon$, cell $C$ contains at most one free space
region $R$, and $R$ is $x$-monotone, $y$-monotone, and connected.

\end{lemma}

\begin{proof}
The monotonicity of $R$ follows from Corollary \ref{corollary-rudolph}.
For connectedness, choose any two free space points $(p_{1},q_{1}),(p_{2},q_{2})$,
and construct a path connecting them in the free space as follows:
move vertically from $(p_{1},q_{1})$ to the minimum point on its
vertical. Do the same for $(p_{2},q_{2})$. By Lemma $\ref{lemma-rudolph}$,
this movement causes the distance to decrease monotonically. By Lemma
\ref{lemma-hourglass-bitonicity}, any two minimum points are connected
by a $\downarrow\uparrow$-bitonic distance function $H_{\overline{ab},\ \overline{cd}}$
(cf.\ section \ref{sub:Funnels-and-Hourglasses}), but as the starting
points are in the free space -- and therefore have distance at most
$\varepsilon$ -- all points on this constructed path lie in the free
space. 
\end{proof}
Given $C$'s boundaries, it is possible to propagate reachability
information (see section \ref{sec:Preliminaries}) through $C$ in
constant time. This follows from the monotonicity and connectedness
of the free space in $C$ and is useful for solving the geodesic decision
problem.

\section{Red-Blue Intersections\label{sub:Red-Blue-Intersections}}

This section shows how to efficiently count and report a certain type
of red-blue intersections in the plane. 
This problem is interesting both from theoretical and applied stances
and will prove useful in section \ref{sub:ParametricSearch} for the
Fréchet optimization problem.

Let $R$ be a set of $m$ {}``red'' curves in the plane such that
every red curve is continuous, $x$-monotone, and monotone \emph{decreasing}.
Let $B$ be a set of $n$ {}``blue'' curves in the plane where each
blue curve is continuous, $x$-monotone, and monotone \emph{increasing}.
Assume that the curves are defined in the slab $[\alpha,\beta]\times\mathbb{R}$,
and let $I(k)$ be the time to find the at most one intersection of
any red and blue curve.\emph{}%
\footnote{There is at most one intersection due to the monotonicities of the
red and blue curves.%
}%

\begin{theorem}\label{Thm: intersectionCounting}

The number of red-blue intersections between $R$ and $B$ in the
slab $[\alpha,\beta]\times\mathbb{R}$ can be \emph{counted} in $O(N\log N)$
total time, where $N=\max(m,n)$. These intersections can be \emph{reported}
in $O(N\log N+K\cdot I(k))$ total time, where $K$ is the total number
of intersections reported. After $O(N\log N)$ preprocessing time,
a \emph{random} red-blue intersection in $[\alpha,\beta]\times\mathbb{R}$
can be returned in $O(\log N+I(k))$ time, and the red curve involved
in the most red-blue intersections can be returned in $O(1)$ time.
All operations require $O(N)$ space.%
\footnote{Palazzi and Snoeyink \cite{Palazzi1994} also count and report red-blue
intersections using a slab-based approach. However, their work is
for line segments instead of curves, and they require that all red
segments are disjoint and all blue segments are disjoint. We have
no such disjointness requirement.%
}

\end{theorem}

\begin{proofsketch}%
Figure \ref{fig:intersectionCounting} illustrates the key idea. Suppose
a red curve $r_{3}(x)$ lies \emph{above} a blue curve $b_{2}(x)$
at $x=\alpha$. If it is also true that $r_{3}(x)$ lies \emph{below}
$b_{2}(x)$ at $x=\beta$, then these monotone curves must intersect
in $[\alpha,\beta]\times\mathbb{R}$. Two sorted lists $L_{\alpha}$,
$L_{\beta}$ of curve values store how many blue curves lie below
each red curve at $x=\alpha$ and $x=\beta$. Subtracting the values
in $L_{\alpha}$ and $L_{\beta}$ yields the number of actual intersections
for each red curve in $[\alpha,\beta]\times\mathbb{R}$ (and also
reveals the red curve that is involved in the most intersections).
Intersection \emph{counting} simply sums up these values. Intersection
\emph{reporting} builds a balanced tree from $L_{\alpha}$ and $L_{\beta}$.%
\begin{figure}
\begin{centering}
\includegraphics[scale=0.43]{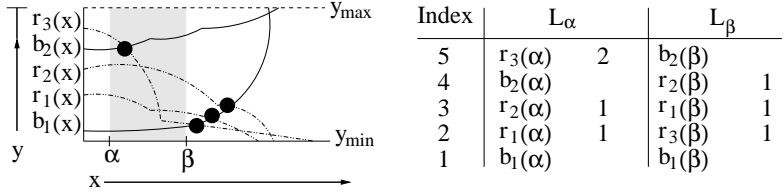} 
\par\end{centering}

\caption{$r_{3}(x)$ lies above \emph{two} blue curves at $x=\alpha$ but only
lies above \emph{one} blue curve at $x=\beta$. Subtraction reveals
that $r_{3}(x)$ has one intersection in the slab $[\alpha,\beta]\times\mathbb{R}$.\label{fig:intersectionCounting}}

\end{figure}

To find a \emph{random }red-blue intersection in $[\alpha,\beta]\times\mathbb{R}$,\emph{
}precompute the number $\kappa$ of red-blue intersections in $[\alpha,\beta]\times\mathbb{R}$.
Pick a random integer between 1 and $\kappa$ and use the number of
intersections stored for each red curve to locate the particular red
curve $r_{i}(x)$ that is involved in the randomly selected intersection.
By searching a persistent version of the reporting structure \cite{Sarnak1986},
$r_{i}(x)$'s $j$th red-blue intersection can be returned in $O(\log N+I(k))$
query time after $O(N\log N)$ preprocessing time. %
{}\end{proofsketch}%

\section{Geodesic Fréchet Algorithm\label{sec:Algorithms}}

\vskip-0.3cm
\subsection{Computing One Cell's Boundaries in $O(\log k)$ Time\label{sub:SecondAlgorithm}}

A boundary of a free space cell is a horizontal (or vertical) line
segment. This boundary can be associated with a funnel $\mathcal{F}_{p,\overline{cd}}$
that has a $\downarrow\uparrow$-bitonic distance function $F_{p,\
\overline{cd}}$ (cf.\ Lemma $\ref{lemma-rudolph}$). Given $\varepsilon\geq0$,
computing the free space on a cell boundary requires finding the (at
most two) values $t_{1}$, $t_{2}$ such that $F_{p,\ \overline{cd}}(t_{1})=F_{p,\
\overline{cd}}(t_{2})=\varepsilon$ (see Figure \ref{fig:cell-boundary}).%
\begin{figure}[H]
\begin{centering}
\includegraphics[scale=0.43]{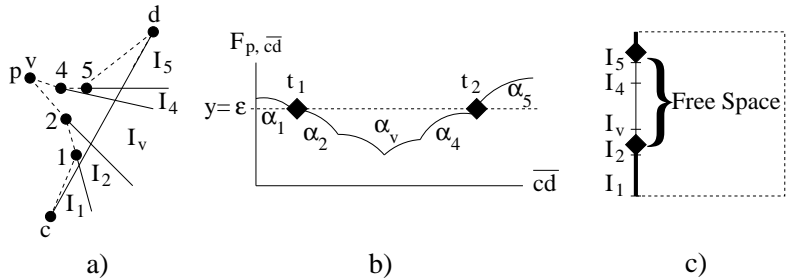} 
\par\end{centering}

\caption{a \& b) A funnel $\mathcal{F}_{p,\ \overline{cd}}$ is associated
with a cell boundary and has a bitonic distance function $F_{p,\ \overline{cd}}$.
c) The (at most two) values $t_{1}$, $t_{2}$ such that $F_{p,\ \overline{cd}}(t_{1})=F_{p,\
\overline{cd}}(t_{2})=\varepsilon$ define the free space on a cell boundary. \label{fig:cell-boundary}}

\end{figure}

\begin{lemma}\label{lemma:intersectionsWithEps}

Both the minimum value of $F_{p,\ \overline{cd}}$ and the (at most
two) values $t_{1}$, $t_{2}$ such that $F_{p,\ \overline{cd}}(t_{1})=F_{p,\
\overline{cd}}(t_{2})=\varepsilon$ can be found for any $\varepsilon\geq0$ in $O(\log k)$ time (after
preprocessing).

\end{lemma}\begin{proofsketch}After $O(k)$ shortest path preprocessing
\cite{Guibas1989,Hershberger1991}, a binary search is performed on
the $O(k)$ arcs of $F_{p,\ \overline{cd}}$ in $O(\log k)$ time.
See our full paper \cite{Cook2007} for details.\end{proofsketch}%
{}\begin{corollary}\label{corollary:cellBoundaries}

The free space on all four boundaries of a free space cell can be
found in $O(\log k)$ time by computing $t_{1}$ and $t_{2}$ for
each boundary.%

\end{corollary}%

\vskip-0.3cm
\subsection{Geodesic Fréchet Decision Problem\label{sub:decisionProblem}}

\begin{theorem}\label{thm: DecisionProblem}

After preprocessing a simple polygon $P$ for shortest path queries
in $O(k)$ time \cite{Guibas1989},%
{} the geodesic Fréchet decision problem for polygonal curves $A$ and
$B$ inside $P$ can be solved for any $\varepsilon\geq0$ in $O(N^{2}\log k)$
time and $O(k+N)$ space.

\end{theorem}

\begin{proof}
Following the standard dynamic programming approach of \cite{Alt1995},
compute all cell boundaries in $O(N^{2}\log k)$ time (cf.\ Corollary
\ref{corollary:cellBoundaries}), and propagate reachability information
through all cells in $O(N^{2})$ time. $O(k)$ space is needed for
the preprocessing structures of \cite{Guibas1989}, and only $O(N)$
space is needed for dynamic programming if two rows of the free space
diagram are stored at a time. 
\end{proof}

\vskip-0.3cm
\subsection{Geodesic Fréchet Optimization Problem\label{sub:ParametricSearch}}

Let $\varepsilon^{*}$ be the minimum value of $\varepsilon$ such
that the Fréchet decision problem returns true. That is, $\varepsilon^{*}$
equals the Fréchet distance $\delta_{F}(A,B)$. Parametric search
is a technique commonly used to find $\varepsilon^{*}$ (see \cite{Agarwal1994,Alt1995,Cole1987,Oostrum2002}).%
\footnote{An easier to implement alternative to parametric search is to run
the decision problem once for every bit of accuracy that is desired.
This approach runs in $O(BN^{2}\log k)$ time and $O(k+N)$ space,
where $B$ is the desired number of bits of accuracy \cite{Oostrum2002}.%
} The typical approach to find $\varepsilon^{*}$ is to sort all the
cell boundary functions based on the unknown parameter $\varepsilon^{*}$.
The comparisons performed during the sort guarantee that the result
of the decision problem is known for all {}``critical values'' \cite{Alt1995}
that could potentially define $\varepsilon^{*}$. Traditionally, such
a sort operates on cell boundaries of constant complexity. The geodesic
case is different because each cell boundary has $O(k)$ complexity.
As a result, a straightforward parametric search based on sorting
these values would require $O(kN^{2}\log kN)$ time even when using
Cole's \cite{Cole1987} optimization.%
\footnote{A variation of the general sorting problem called the {}``nuts and
bolts'' problem (see \cite{Komlos}) is tantalizingly close to an
acceptable $O(N^{2}\log N)$ sort but does not apply to our setting.%
}

We present a randomized algorithm with expected runtime $O(k+N^{2}\log kN\log N)$
and worst-case runtime $O(k+N^{3}\log kN)$. This algorithm is an
order of magnitude faster than parametric search in the expected case.%
\begin{figure}
\begin{centering}
\includegraphics[scale=0.43]{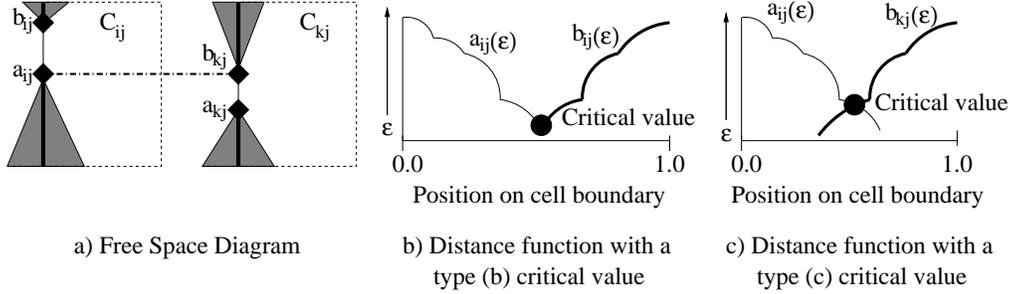} 
\par\end{centering}

\caption{Critical values of the Fréchet distance \label{fig:ParametricSearch}}

\end{figure}

Each cell boundary has at most one free space interval (cf.\ Lemma
\ref{lemma-rudolph}). The upper boundary of this interval is a function
$b_{ij}(\varepsilon)$, and the lower boundary of this interval is
a function $a_{ij}(\varepsilon)$. See Figure \ref{fig:ParametricSearch}a.
The seminal work of Alt and Godau \cite{Alt1995} defines three types
of critical values that are useful for computing the exact geodesic
Fréchet distance. There are exactly two type (a) critical values associated
with distances between the starting points of $A$ and $B$ and the
ending points of $A$ and $B$. Type (b) critical values occur $O(N^{2})$
times when $a_{ij}(\varepsilon)=b_{ij}(\varepsilon)$. See Figure
\ref{fig:ParametricSearch}b. Type (a) and (b) critical values occur
$O(N^{2})$ times and are easily handled in $O(N^{2}\log k\log N)$
time. This process involves computing values in $O(N^{2}\log k)$
time, sorting in $O(N^{2}\log N)$ time, and running the decision
problem in binary search fashion $O(\log N)$ times. %
{}Resolving the type (a) and (b) critical values as a first step will
{} simplify the randomized algorithm for the type (c) critical values.

Alt and Godau \cite{Alt1995} show that type (c) critical values occur
when the position of $a_{ij}(\varepsilon)$ in cell $C_{ij}$ equals
the position of $b_{kj}(\varepsilon)$ in cell $C_{kj}$ in the free
space diagram. See Figure \ref{fig:ParametricSearch}a. As $\varepsilon$
increases, by Lemma \ref{lemma-rudolph}, $a_{ij}(\varepsilon)$ is
$\downarrow$-monotone on the cell boundary and $b_{ij}(\varepsilon)$
is $\uparrow$-monotone (see Figure \ref{fig:ParametricSearch}b).
As illustrated in Figure \ref{fig:ParametricSearch}c, $a_{ij}(\varepsilon)$
and $b_{kj}(\varepsilon)$ intersect at most once. This follows from
the monotonicities of $a_{ij}(\varepsilon)$ and $b_{kj}(\varepsilon)$.
Hence, there are $O(N^{2})$ intersections of $a_{ij}(\varepsilon)$
and $b_{kj}(\varepsilon)$ in row $j$ and a total of $O(N^{3})$
type (c) critical values over all rows. There are also $O(N^{2})$
intersections of $a_{ij}(\varepsilon)$ and $b_{ik}(\varepsilon)$
in \emph{column} $i$ and a total of $O(N^{3})$ additional type (c)
critical values over all columns.

\begin{lemma}\label{lemma:arcIntersections}

The intersection of $a_{ij}(\varepsilon)$ and $b_{kl}(\varepsilon)$
can be found for any $\varepsilon\geq0$ in $O(\log k)$ time after
preprocessing.

\end{lemma}\begin{proofsketch}Build binary search trees for $a_{ij}(\varepsilon)$
and $b_{kl}(\varepsilon)$ and perform a binary search. See our full
paper \cite{Cook2007} for details.\end{proofsketch}

{}Theorem \ref{Thm: intersectionCounting} requires that all $a_{ij}(\varepsilon)$
and $b_{kl}(\varepsilon)$ are defined in the slab $[\alpha,\beta]\times\mathbb{R}$
that contains $\varepsilon^{*}$. Precomputing the type (a) and type
(b) critical values of \cite{Alt1995} shrinks the slab such that
no \emph{left} endpoint of any relevant $a_{ij}(\varepsilon)$, $b_{kl}(\varepsilon)$
appears in $[\alpha,\beta]\times\mathbb{R}$ when processing the type
(c) critical values. In addition, $a_{ij}(\varepsilon)$, $b_{kl}(\varepsilon)$
can be extended horizontally so that no \emph{right} endpoint appears
in $[\alpha,\beta]\times\mathbb{R}$. These changes do not affect
the asymptotic number of intersections and allow Theorem \ref{Thm: intersectionCounting}
to count and report type (c) critical values in $[\alpha,\beta]\times\mathbb{R}$.%

The below randomized algorithm solves the geodesic Fréchet optimization
problem in $O(k+N^{2}\log kN\log N)$ expected time. This is faster
than the standard parametric search approach which requires $O(kN^{2}\log kN)$
time.\\

\noindent \textbf{Randomized Optimization Algorithm}

\begin{enumerate}
\item Precompute and sort all type (a) and type (b) critical values in $O(N^{2}\log kN)$
time (cf.\ Lemma \ref{lemma:intersectionsWithEps}). Run the decision
problem $O(\log N)$ times to resolve these values and shrink the
potential slab for $\varepsilon^{*}$ down to $[\alpha,\beta]\times\mathbb{R}$
in $O(N^{2}\log k\log N)$ time. 
\item \label{enu:StepCounting}Count the number $\kappa_{j}$ of type (c)
critical values for each row $j$ in the slab $[\alpha,\beta]\times\mathbb{R}$
using Theorem \ref{Thm: intersectionCounting}. Let $C_{j}$ be the
resulting counting data structure for row $j$. 
\item To achieve a fast \emph{expected} runtime, pick a random intersection
$\vartheta_{j}$ for each row using $C_{j}$.%
\footnote{Picking a critical value at random is related to the distance selection
problem \cite{Bespamyatnikh2004} and is mentioned in \cite{Agarwal1998a},
but to our knowledge, this alternative to parametric search has never
been applied to the Fréchet distance.%
} See Theorem \ref{Thm: intersectionCounting}.
\item To achieve a fast \emph{worst-case} runtime, use $C_{j}$ to find
the $a_{Mj}(\varepsilon)$ curve in each row that has the most intersections
(see Theorem \ref{Thm: intersectionCounting}). Add all intersections
in $[\alpha,\beta]\times\mathbb{R}$ that involve $a_{Mj}(\varepsilon)$
to a global pool $\mathcal{P}$ of unresolved critical values%
\footnote{The idea of a global pool is similar to Cole's optimization for parametric
search \cite{Cole1987}.%
} and delete $a_{Mj}(\varepsilon)$ from any future consideration. 
\item \label{enu:globalPoolOperations}Find the median $\Xi$ of the values
in $\mathcal{P}$ in $O(N^{2})$ time using the standard median algorithm
mentioned in \cite{Komlos}. Also find the median $\Psi$ of the $O(N)$
randomly selected $\vartheta_{j}$ in $O(N)$ time using a \emph{weighted}
median algorithm based on the number of critical values $\kappa_{j}$
for each row $j$.
\item \label{enu:StepDecisionPblm}Run the decision problem twice: once
on $\Xi$ and once on $\Psi$. This shrinks the search slab $[\alpha,\beta]\times\mathbb{R}$
and \emph{at least} halves the size of $\mathcal{P}$. Repeat steps
\ref{enu:StepCounting} through \ref{enu:StepDecisionPblm} until
all \emph{row}-based type (c) critical values have been resolved. 
\item Resolve all \emph{column}-based type (c) critical values in the same
spirit as steps \ref{enu:StepCounting} through \ref{enu:StepDecisionPblm}
and return the smallest critical value that satisfied the decision
problem as the value of the geodesic Fréchet distance. 
\end{enumerate}
\begin{theorem}\label{theorem-geodesic}

The exact \emph{geodesic} Fréchet distance between two polygonal curves
$A$ and $B$ inside a simple bounding polygon $P$ can be computed
in $O(k+N^{2}\log kN\log N)$ expected time and $O(k+N^{3}\log kN)$
worst-case time, where $N$ is the larger of the complexities of $A$
and $B$ and $k$ is the complexity of $P$. $O(k+N^{2})$ space is
required.\end{theorem}

\begin{proof}
Preprocess $P$ once for shortest path queries in $O(k)$ time \cite{Guibas1989}.
In the expected case, each execution of the decision problem will
eliminate a constant fraction of the remaining type (c) critical values
due to the proof of Quicksort's expected runtime and the median of
medians approach for $\Psi$. Consequently, the expected number of
iterations of the algorithm is $O(\log N^{3})=O(\log N)$.

In the worst-case, each of the $O(N)$ $a_{ij}(\varepsilon)$ in a
row will be picked as $a_{Mj}(\varepsilon)$. Therefore, each row
can require at most $O(N)$ iterations. Since \emph{all} rows are
processed each iteration, the entire algorithm requires at most $O(N)$
iterations for \emph{row}-based critical values. By a similar argument,
\emph{column}-based critical values also require at most $O(N)$ iterations.

The size of the pool $\mathcal{P}$ is expressed by the inequality
$S(x)\leq\frac{S(x-1)+O(N^{2})}{2}$, where $x$ is the current step
number, and $S(0)=0$. Intuitively, each step adds $O(N^{2})$ values
to $\mathcal{P}$ and then at least half of the values in $\mathcal{P}$
are always resolved using the median $\Xi$. It is not difficult to
show that $S(x)\in O(N^{2})$ for any step number $x$.

Each iteration of the algorithm requires intersection counting and
intersection calculations for $O(N)$ rows (or columns) at a cost
of $O(N^{2}\log kN)$ time. In addition, the global pool $\mathcal{P}$
has its median calculated in $O(N^{2})$ time, and the decision problem
is executed in $O(N^{2}\log k)$ time. Consequently, the expected
runtime is $O(k+N^{2}\log kN\log N)$ and the worst-case runtime is
$O(k+N^{3}\log kN)$ including $O(k)$ preprocessing time \cite{Guibas1989}
for geodesics. The preprocessing structures use $O(k)$ space that
must remain allocated throughout the algorithm, and the pool $\mathcal{P}$
uses $O(N^{2})$ additional space. 
\end{proof}
{}Although the exact non-geodesic Fréchet distance is normally found
in $O(N^{2}\log N)$ time using parametric search (see \cite{Alt1995}),
parametric search is often regarded as impractical because it is difficult
to implement%
\footnote{Quicksort-based parametric search has been implemented by van Oostrum
and Veltkamp \cite{Oostrum2002} using a complex framework.%
} and involves enormous constant factors \cite{Cole1987}. To the best
of our knowledge, the randomized algorithm in section \ref{sub:ParametricSearch}
provides the first practical alternative to parametric search for
solving the exact non-geodesic Fréchet optimization problem in $\mathbb{R}^{l}$.

\begin{theorem} \label{theorem-non-geodesic-frechet}

The exact \emph{non-geodesic} Fréchet distance between two polygonal
curves $A$ and $B$ in $\mathbb{R}^{l}$ can be computed in $O(N^{2}\log^{2}N)$
expected time, where $N$ is the larger of the complexities of $A$
and $B$. $O(N^{2})$ space is required.

\end{theorem}

\begin{proof}
The argument is very similar to the proof of Theorem \ref{theorem-geodesic}.
The main difference is that non-geodesic distances can be computed
in $O(1)$ time (instead of $O(\log k)$ time).
\end{proof}

\section{Geodesic Fréchet Distance in a Polygonal Domain with Obstacles\label{sec:continuousGeodesicFrechet}}

Consider the real-life situation of a person walking a dog in a park.
If the person and dog walk on opposite sides of a group of trees,
then the leash must go around the trees. 
More formally, suppose the two polygonal curves $A$ and $B$ lie
in a planar polygonal domain $\mathcal{D}$ \cite{Mitchell1998} of
complexity $k$. The leash is required to change continuously, i.e.,
it must stay inside $\mathcal{D}$ and may not pass through or jump
over an obstacle. It may, however, cross itself. Let $\delta_{C}$
be the geodesic Fréchet distance for this scenario when the leash
length is measured geodesically.%
\footnote{We recently learned that this topic has been independently explored
in \cite{Chambers2007}.%
}

Due to the continuity of the leash's motion, the free space inside
a geodesic cell is represented by an hourglass -- just as it was for
the geodesic Fréchet distance inside a simple polygon. Hence, free
space in a cell is $x$-monotone, $y$-monotone, and connected (cf.\ Lemma
\ref{one-free-region-per-cell}), and reachability information can
be propagated through a cell in constant time. 

The main task in computing $\delta_{C}$ is to construct all cell
boundaries. Once the cell boundaries are known, the decision and optimization
problems can be solved by the algorithms for the geodesic Fréchet
distance inside a simple polygon (cf.\ Theorems \ref{thm: DecisionProblem}
and \ref{theorem-geodesic}). We use Hershberger and Snoeyink's homotopic
shortest paths algorithm \cite{Hershberger1991} to incrementally
construct all cell boundary funnels needed to compute $\delta_{C}$.
To use the homotopic algorithm, the polygonal domain $\mathcal{D}$
should be triangulated in $O(k\log k)$ time \cite{Mitchell1998},
and all obstacles should be replaced by their vertices. A shortest
path map \cite{Mitchell1998} can find an initial geodesic leash $L_{I}$
between the start points of the polygonal curves $A$ and $B$ in
$O(k\log k)$ time.

\begin{lemma}\label{lemma-continuousCells}

Given the initial leash for the bottom-left corner of a $\delta_{C}$-cell
$\mathcal{C}$, all four funnel boundaries of $\mathcal{C}$ and the
initial leashes for cells adjacent to $\mathcal{C}$ can be computed
in $O(k)$ time.

\end{lemma}

\begin{proof}
The funnels representing cell boundaries are constructed \emph{incrementally}.
The idea is to extend the initial leash into a homotopic {}``sketch''
that describes how the shortest path should wind through the obstacles
and then to {}``snap'' this sketch into a shortest path (see Figures
\ref{fig:continuousFunnel}a and \ref{fig:continuousFunnel}b).%
\begin{figure}
\begin{centering}
\includegraphics[scale=0.4]{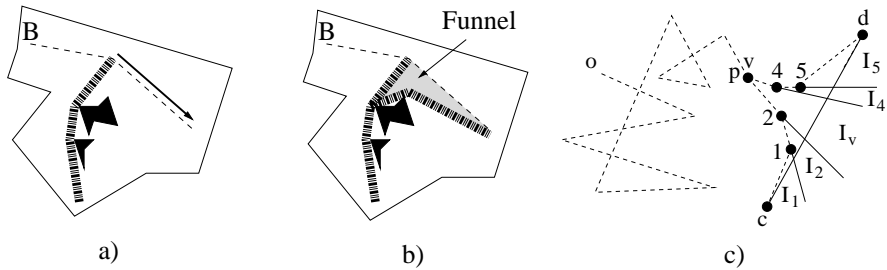} 
\par\end{centering}
\vskip-0.2cm
\caption{a) A funnel for a $\delta_{C}$-cell can be found by extending a cell's
initial leash along one segment to create a path sketch and then b)
snapping this sketch into a homotopic shortest path. c) A funnel $\mathcal{F}_{o,\ \overline{cd}}$
has $O(kN)$ complexity, but the distance function $F_{o,\ \overline{cd}}$
has only $O(k)$ complexity because $d(o,p)$ is a constant.\label{fig:continuousFunnel}}

\end{figure}

Homotopic shortest paths have increased complexity over normal shortest
paths because they can loop around obstacles.%
{} For example, if the person walks in a triangular path around all
the obstacles, then the leash follows a homotopic shortest path that
can have $O(k)$ complexity in a single cycle around the obstacles.
By repeatedly winding around the obstacles $O(N)$ times, a path achieves
$O(kN)$ complexity.

To avoid spending $O(kN)$ time per cell, we extend a previous homotopic
shortest path into a sketch by appending a single line segment to
the previous path (see Figure \ref{fig:continuousFunnel}a). Adding
this single segment can unwind at most one loop over a subset of obstacles,
so only the most recent $O(k)$ vertices of the sketch will need to
be updated when the sketch is snapped into the true homotopic shortest
path. A turning angle is used to identify these $O(k)$ vertices by
backtracking on the sketch until the angle is at least $2\pi$ different
from the final angle. %

Putting all this together, a boundary for a free space cell can be
computed in $O(k)$ time by starting with an initial leash $L_{I}$
of $O(kN)$ complexity, constructing a homotopic sketch by appending
a single segment to $L_{I}$, backtracking with a turning angle to
find $O(k)$ vertices that are eligible to be changed, and finally
{}``snapping'' these $O(k)$ vertices to the true homotopic shortest
path using Hershberger and Snoeyink's algorithm \cite{Hershberger1991}.
The result is a funnel that describes one cell boundary.

By extending $L_{I}$ in four combinatorially distinct ways, all four
cell boundaries can be defined. Specifically, we can extend $L_{I}$
along the current $\overline{ab}\in A$ segment to form the first
funnel or along the $\overline{cd}\in B$ segment to form the second
funnel. The third funnel is created by extending $L_{I}$ along $\overline{ab}\in A$
and then $\overline{cd}\in B$. The fourth funnel is created by extending
$L_{I}$ along $\overline{cd}\in B$ and then $\overline{ab}\in A$.
These cell boundaries conveniently define the initial leash for cells
that are adjacent to $\mathcal{C}$.
\end{proof}
\begin{theorem}\label{theorem-continuousDecisionProblem}

The $\delta_{C}$ decision problem can be solved in $O(kN^{2})$ time
and $O(k+N)$ space.

\end{theorem}

\begin{proof}
Each cell boundary is a funnel $\mathcal{F}_{o,\ \overline{cd}}$
with $O(kN)$ complexity \cite{Duncan2006}. However, this high complexity
is a result of looping over obstacles, and most of these points do
not affect the funnel's distance function $F_{o,\ \overline{cd}}$.
As illustrated in Figure \ref{fig:continuousFunnel}c, $F_{o,\ \overline{cd}}$
has only $O(k)$ complexity because only vertices $\pi(p,c)\cup\pi(p,d)$
contribute arcs to $F_{o,\ \overline{cd}}$.

Construct all cell boundary funnels in $O(kN^{2})$ time (cf.\ Lemma
\ref{lemma-continuousCells}), intersect each funnel's distance function
with $y=\varepsilon$ in $O(N^{2}\log k)$ time, and propagate reachability
information in $O(N^{2})$ time. Only $O(k+N)$ space is needed for
dynamic programming when storing only two rows at a time. 
\end{proof}
\begin{theorem}\label{theorem-continuousOptimizationProblem}

The $\delta_{C}$ optimization problem can be solved in $O(kN^{2}+N^{2}\log kN\log N)$
expected time and $O(kN^{2})$ space.%
\footnote{If space is at a premium, the algorithm can also run with $O(k+N^{2})$
space and $O(kN^{2}\log N+N^{2}\log kN\log N)$ expected time by recomputing
the funnels each time the decision problem is computed. Note that
$O(N^{2})$ storage is required for the red-blue intersections algorithm
(cf.\ Theorem \ref{theorem-geodesic}).%
}

\end{theorem}

\begin{proof}
The $\delta_{C}$ optimization problem can be solved using red-blue
intersections. $O(\log N)$ steps are performed in the expected case
by Theorem \ref{theorem-geodesic}. Each step has to perform intersection
counting in $O(N^{2}\log kN)$ time and solve the decision problem.
If the funnels are precomputed in $O(kN^{2})$ time and space, then
the decision problem can be solved in $O(N^{2}\log k)$ time. Hence,
after $O(kN^{2})$ time and space preprocessing, $\delta_{C}$ can
be found in $O(\log N)$ expected steps where each step takes $O(N^{2}\log kN)$
time.
\end{proof}

\vskip-0.3cm
\section{Geodesic Hausdorff Distance\label{sec:Geodesic-Hausdorff-Distance}}

Hausdorff distance is a similarity metric commonly used to compare
sets of points or sets of line segments. The \emph{directed} geodesic
Hausdorff distance can be formally defined as $\tilde{\delta}_{H}(A,B)=\ \sup_{a\in A}\inf_{b\in B}d(a,b)$,
where $A$ and $B$ are sets and $d(a,b)$ is the geodesic distance
between $a$ and $b$ (see \cite{Alt1995,Alt2003}). The \emph{undirected}
geodesic Hausdorff distance is the larger of the two directed distances:
$\delta_{H}(A,B)=\max(\tilde{\delta}_{H}(A,B),\ \tilde{\delta}_{H}(B,A))$.

\begin{theorem}\label{theorem-geodesic-hausdorff}

$\delta_{H}(A,B)$ for point sets $A,B$ inside a simple polygon $P$
can be computed in $O((k+N)\log(k+N))$ time and $O(k+N)$ space,
where $N$ is the larger of the complexities of $A$ and $B$ and
$k$ is the complexity of $P$. If $A$ and $B$ are sets of line
segments, $\delta_{H}(A,B)$ can be computed in $O(kN^{2}\alpha(kN)\log kN)$
time and $O(kN\alpha(kN)\log kN)$ space.

\end{theorem}

\begin{proofsketch}A geodesic Voronoi diagram \cite{Papadopoulou1998}
finds nearest neighbors when $A$ and $B$ are point sets. When $A$
and $B$ are sets of line segments, all nearest neighbors for a line
segment can be found by computing a lower envelope \cite{Agarwal1995}
of $O(N)$ hourglass distance functions. The largest nearest neighbor
distance over all line segments is $\delta_{H}(A,B)$. %
{}\end{proofsketch}

\section{Conclusion}

To compute the geodesic Fréchet distance between two polygonal curves
inside a simple polygon, we have proven that the free space inside
a geodesic cell is $x$-monotone, $y$-monotone, and connected. By
extending the shortest path algorithms of \cite{Guibas1989,Hershberger1991},
the boundaries of a single free space cell can be computed in logarithmic
time, and this leads to an efficient algorithm for the geodesic Fréchet
decision problem.

A randomized algorithm based on red-blue intersections solves the
geodesic Fréchet optimization problem in lieu of the standard parametric
search approach. The randomized algorithm is also a practical alternative
to parametric search for the non-geodesic Fréchet distance in arbitrary
dimensions.

We can compute the geodesic Fréchet distance between two polygonal
curves $A$ and $B$ inside a simple bounding polygon $P$ in $O(k+N^{2}\log kN\log N)$
expected time, where $N$ is the larger of the complexities of $A$
and $B$ and $k$ is the complexity of $P$. In the expected case,
the randomized optimization algorithm is an order of magnitude faster
than a straightforward parametric search that uses Cole's \cite{Cole1987}
optimization to sort $O(kN^{2})$ values.

The geodesic Fréchet distance in a polygonal domain with obstacles
enforces a homotopy on the leash. It can be computed in the same manner
as the geodesic Fréchet distance inside a simple polygon after computing
cell boundary funnels using Hershberger and Snoeyink's homotopic shortest
paths algorithm \cite{Hershberger1991}. Future work could attempt
to compute these funnels in $O(\log k)$ time instead of $O(k)$ time.
The geodesic Hausdorff distance for point sets inside a simple polygon
can be computed using geodesic Voronoi diagrams. The geodesic Hausdorff
distance for line segments can be computed using lower envelopes;
future work could speed up this algorithm by developing a geodesic
Voronoi diagram for line segments.


\begin{thebibliography}{10}

\bibitem{Agarwal1995}
P.~K. Agarwal and M.~Sharir.
\newblock Davenport--{S}chinzel sequences and their geometric applications.
\newblock Technical Report Technical report DUKE--TR--1995--21, 1995.

\bibitem{Agarwal1998a}
P.~K. Agarwal and M.~Sharir.
\newblock Efficient algorithms for geometric optimization.
\newblock {\em ACM Comput. Surv.}, 30(4):412--458, 1998.

\bibitem{Agarwal1994}
P.~K. Agarwal, M.~Sharir, and S.~Toledo.
\newblock Applications of parametric searching in geometric optimization.
\newblock volume~17, pages 292--318, Duluth, MN, USA, 1994. Academic Press,
  Inc.

\bibitem{Alt1995}
H.~Alt and M.~Godau.
\newblock Computing the {F}réchet distance between two polygonal curves.
\newblock {\em International Journal of Computational Geometry and
  Applications}, 5:75--91, 1995.

\bibitem{Alt2003}
H.~Alt, C.~Knauer, and C.~Wenk.
\newblock Comparison of distance measures for planar curves.
\newblock {\em Algorithmica}, 38(1):45--58, 2003.

\bibitem{Bespamyatnikh2004}
S.~Bespamyatnikh and M.~Segal.
\newblock Selecting distances in arrangements of hyperplanes spanned by points.
\newblock volume~2, pages 333--345, September 2004.

\bibitem{Buchin2006}
K.~Buchin, M.~Buchin, and C.~Wenk.
\newblock Computing the {F}réchet distance between simple polygons in
  polynomial time.
\newblock {\em SoCG: 22nd Symposium on Computational Geometry}, pages 80--87,
  2006.

\bibitem{Chambers2007}
E.~W. Chambers, É.~C. de~Verdière, J.~Erickson, S.~Lazard, F.~Lazarus, and
  S.~Thite.
\newblock Walking your dog in the woods in polynomial time.
\newblock {\em 17th Fall Workshop on Computational Geometry}, 2007.

\bibitem{Cole1987}
R.~Cole.
\newblock Slowing down sorting networks to obtain faster sorting algorithms.
\newblock {\em J. ACM}, 34(1):200--208, 1987.

\bibitem{Cook2007}
A.~F. {Cook IV} and C.~Wenk.
\newblock Geodesic {F}réchet and {H}ausdorff distance inside a simple polygon.
\newblock Technical Report CS-TR-2007-004, University of Texas at San Antonio,
  August 2007.

\bibitem{Duncan2006}
C.~A. Duncan, A.~Efrat, S.~G. Kobourov, and C.~Wenk.
\newblock Drawing with fat edges.
\newblock {\em Int. J. Found. Comput. Sci.}, 17(5):1143--1164, 2006.

\bibitem{Efrat2002}
A.~Efrat, L.~J. Guibas, S.~Har-Peled, J.~S.~B. Mitchell, and T.~M. Murali.
\newblock New similarity measures between polylines with applications to
  morphing and polygon sweeping.
\newblock {\em Discrete {\&} Computational Geometry}, 28(4):535--569, 2002.

\bibitem{Guibas1989}
L.~J. Guibas and J.~Hershberger.
\newblock Optimal shortest path queries in a simple polygon.
\newblock {\em J. Comput. Syst. Sci.}, 39(2):126--152, 1989.

\bibitem{Guibas1986}
L.~J. Guibas, J.~Hershberger, D.~Leven, M.~Sharir, and R.~E. Tarjan.
\newblock Linear time algorithms for visibility and shortest path problems
  inside simple polygons.
\newblock pages 1--13, 1986.

\bibitem{Guibas1987}
L.~J. Guibas, J.~Hershberger, D.~Leven, M.~Sharir, and R.~E. Tarjan.
\newblock Linear-time algorithms for visibility and shortest path problems
  inside triangulated simple polygons.
\newblock {\em Algorithmica}, 2:209--233, 1987.

\bibitem{Hershberger1991}
J.~Hershberger.
\newblock A new data structure for shortest path queries in a simple polygon.
\newblock {\em Inf. Process. Lett.}, 38(5):231--235, 1991.

\bibitem{Komlos}
J.~Koml\'{o}s, Y.~Ma, and E.~Szemer\'{e}di.
\newblock Matching nuts and bolts in {O}(n log n) time.
\newblock {\em SODA: 7th ACM-SIAM Symposium on Discrete Algorithms}, pages
  232--241, 1996.

\bibitem{Maheshwari2005}
A.~Maheshwari and J.~Yi.
\newblock On computing {F}réchet distance of two paths on a convex polyhedron.
\newblock {\em EWCG 2005}, pages 41--4, 2005.

\bibitem{Mitchell1998}
J.~S.~B. Mitchell.
\newblock Geometric shortest paths and network optimization.
\newblock {\em Handbook of Computational Geometry}, 1998.

\bibitem{Mitchell1987}
J.~S.~B. Mitchell, D.~M. Mount, and C.~H. Papadimitriou.
\newblock The discrete geodesic problem.
\newblock {\em SIAM J. Comput.}, 16(4):647--668, 1987.

\bibitem{Palazzi1994}
L.~Palazzi and J.~Snoeyink.
\newblock Counting and reporting red/blue segment intersections.
\newblock {\em CVGIP: Graph. Models Image Process.}, 56(4):304--310, 1994.

\bibitem{Papadopoulou1998}
E.~Papadopoulou and D.~T. Lee.
\newblock A new approach for the geodesic {V}oronoi diagram of points in a
  simple polygon and other restricted polygonal domains.
\newblock {\em Algorithmica}, 20(4):319--352, 1998.

\bibitem{Rote2005}
G.~Rote.
\newblock Computing the {F}r\'echet distance between piecewise smooth curves.
\newblock Technical Report ECG-TR-241108-01, May 2005.

\bibitem{Sarnak1986}
N.~Sarnak and R.~E. Tarjan.
\newblock Planar point location using persistent search trees.
\newblock {\em Commun. ACM}, 29(7):669--679, 1986.

\bibitem{Oostrum2002}
R.~van Oostrum and R.~C. Veltkamp.
\newblock Parametric search made practical.
\newblock {\em SoCG: 18th Symposium on Computational Geometry}, pages 1--9,
  2002.

\bibitem{Wenk2006}
C.~Wenk, R.~Salas, and D.~Pfoser.
\newblock Addressing the need for map-matching speed: Localizing global
  curve-matching algorithms.
\newblock {\em 18th Int'l Conf. on Sci. and Statistical Database Mgmt (SSDBM)},
  pages 379--388, 2006.

\end{thebibliography}
\end{document}